\documentclass{article}

\usepackage[utf8]{inputenc}

\usepackage{algorithm}

\usepackage{amsmath}
\usepackage{algpseudocode}
\usepackage{blindtext}
\usepackage[margin=1in, a4paper, total={6.0in, 8in}]{geometry}

\usepackage{graphicx}
\usepackage{caption}
\usepackage{subcaption}
\usepackage[maxbibnames=99]{biblatex}
\DeclareUnicodeCharacter{0301}{}
\DeclareUnicodeCharacter{2113}{$\ell$}

\usepackage{amsthm}
\newtheorem{theorem}{Theorem}[section]
\newtheorem{lemma}[theorem]{Lemma}
\newtheorem{corollary}[theorem]{Corollary}

\title{Fast and Practical DAG Decomposition with Reachability Applications}

\author{Giorgos Kritikakis\\ georgecretek@gmail.com
\\
\and Ioannis G. Tollis \\ tollis@csd.uoc.gr  
\and Computer Science Department, University of Crete, GREECE %
}

\date{
June 2022 - Updated November 2022
}

\begin{document}

\maketitle
\begin{abstract}

We present practical linear and almost linear-time algorithms to compute a chain decomposition of a \emph{directed acyclic graph} (DAG), $G=(V,E)$. The number of vertex-disjoint chains computed
is very close to the minimum. The time complexity of our algorithm is $O(|E|+c*l)$, where $c$ is the number of path concatenations and $l$ is the length of a longest path of the graph. We give a comprehensive explanation on factors $c$ and $l$ in the following sections. 
Our techniques have important applications in many areas, including the design of faster practical transitive closure algorithms. We observe that $|E_{red}|\leq width*|V|$ ($E_{red}$: non-transitive edges) and show how to find a substantially large subset of $E_{tr}$ (transitive edges) using a chain decomposition in linear time, without calculating the transitive closure. Our extensive experimental results show the interplay between the width, $E_{red}$, $E_{tr}$ in various models of graphs.
We show how to compute a reachability indexing scheme in $O(k_c*|E_{red}|)$ time, where $k_c$ is the number of chains and $|E_{red}|$ is the number of non-transitive edges. This scheme can answer reachabilitiy queries in constant time. The space complexity of the scheme is $O(k_c*|V|)$.  
The experimental results reveal that our methods are even better in practice than the theoretical bounds imply,  indicating how fast chain decomposition algorithms can be applied to the transitive closure problem.
\\
\\

\textbf{Keywords:} graph algorithms, hierarchy, directed acyclic graphs (DAG), path/chain decomposition, transitive closure, transitive reduction, reachability, reachability indexing scheme.

\end{abstract}




\section{Introduction}



\emph{Directed acyclic graphs} (DAGs) are very important in many applications. The \emph{width} of a DAG $G=(V,E)$ is the maximum number of mutually unreachable vertices of $G$~\cite{DILWORTH}. 
The width of a DAG is a crucial metric in graph theory and it is used in a wide area of applications. Any DAG can be decomposed into vertex disjoint \emph{paths} or \emph{chains}. In a path every vertex is connected to its
successor by an edge, while in a chain any vertex is connected to its successor by a directed path, which may be an edge. The computation of the width and path/chain decomposition have many applications including bioinformatics \cite{bonizzoni2007linear,gramm2007haplotyping}, evolutionary computation \cite{jaskowski2011formal}, databases \cite{Jagadish:1990:CTM:99935.99944},
graph drawing \cite{ortali2018algorithms,NewFrHierDrawings,lionakis2020algorithms}, distributed systems \cite{ikiz2006efficient,tomlinson1997monitoring}.

An optimum \emph{chain decomposition} of a DAG $G=(V,E)$ contains the minimum number of chains, $k$, which is equal to the width of $G$.  Due to the multitude of applications there are several algorithms to find an optimum chain decomposition, see for example~\cite{Jagadish:1990:CTM:99935.99944,Fulkerson,chendag,makinen2019sparse,caceres2021linear,caceres2022sparsifying,kogan2022beating,van2021minimum}. An FPT algorithm was presented in~\cite{caceres2021linear} that runs in $O(k^2 4^k |V | + k2^k |E|)$ time.  Clearly, this is practical only for very specific classes of DAGs that have very small values of $k$.  Most DAGs have much higher width (often a function of $n$), as shown in the experimental results of Section~\ref{sec:ChainDec}. 
Recently, better almost-linear-time algorithms were presented in~\cite{caceres2022sparsifying,kogan2022beating,van2021minimum} whose time complexity is $O(k^3|V | + |E|)$ and $\tilde{O}(|E| + |V| ^ {3/2})$, respectively.  All of these algorithms solve the chain decomposition problem by reducing it to a minimum cost flow problem.  Using the best minimum cost flow algorithm, which runs in almost-linear time~\cite{van2021minimum}, solves the problem theoretically.
However, due to the heavy mechanism, this approach is challenging to implement. 
Additionally, for several applications it is not necessary to compute an optimum chain decomposition.  We will describe heuristic techniques that compute a chain decomposition that is close to the optimum in linear or almost linear time and are easy to implement. We present a simple greedy algorithm for the chain decomposition problem.
We perform extensive experiments on several random graphs to verify its effectiveness and explore how fast chain decomposition can enhance transitive closure solutions.  Our experimental results show that the algorithm produces decompositions whose sizes are close to optimal.

In \cite{Jagadish:1990:CTM:99935.99944}, Jagadish categorized path and chain decomposition heuristics into two kinds, Chain-Order Heuristic and Node-Order Heuristic. The Chain-Order Heuristic constructs the paths one by one, while the Node-Order Heuristic creates the paths in parallel. Jagadish's heuristics  for path decomposition run in linear time, while his chain decomposition heuristics run in $O(n^2)$ time given a precomputed transitive closure. Few techniques have been presented to compute sub-optimal chain decomposition since then, see~\cite{chen2008efficient,boria2016greedy}, but none of them is as simple and fast as the algorithm we describe here.

In this paper, we explore new path and chain decomposition algorithms for DAGs.
Our focus is on practical algorithms that run in linear or almost linear time, and produce results that are very close to the optimum.
An interesting observation is that the width of some DAGs follows the curve $width =\frac{\mbox{nodes}}{\mbox{average degree}}$.  We apply the computed decomposition in order to obtain several techniques that solve fundamental problems, such as transitive closure, faster than previous algorithms. We use the notion of chain decomposition to offer bounds to the number of non transitive edges and explore how it facilitates in various transitive closure and reduction problems.  Clearly, our decomposition techniques can be used to solve various other problems in order to obtain new time and space complexity bounds.


In Section \ref{sec:PathDec}, we present path decomposition algorithms. In Section \ref{sec:ChainDec}, chain decomposition and path concatenation approaches are presented. Additionally, we show experiments and evaluate the performance of our heuristics. Moreover, the value of the width is explored as the graph becomes denser.
Next, we examine some applications of the heuristics. In Section \ref{sec:HierTrans}, we show that $|E_{red}| \leq width*|V|$, and describe how to remove a significant subset of transitive edges, $E_{tr}'$, in linear time, in order to bound $|E-E_{tr}'|$ by $O(k*V)$ given any path/chain decomposition of size $k$. This can boost many known transitive closure solutions. Finally, Section \ref{sec:IndexingScheme} demonstrates how to build an indexing scheme that implicitly contains the transitive closure and we report experimental results. Our experiments shed light on the behavior of critical factors. 

All experiments were conducted on a laptop PC (Intel(R) Core(TM) i5-6200U CPU, with 8 GB of main memory). Our algorithms have been developed as stand-alone java programs.
\begin{figure}
     \centering
     \begin{subfigure}[b]{0.19 \textwidth}
         \centering
         \includegraphics[width=\textwidth]{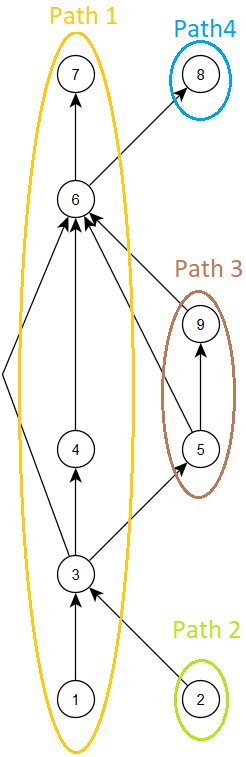}
         \caption{A path decomposition of a graph consisting of 4 paths.}
     \end{subfigure}
     \hfill
     \begin{subfigure}[b]{0.19 \textwidth}
         \centering
         \includegraphics[width=\textwidth]{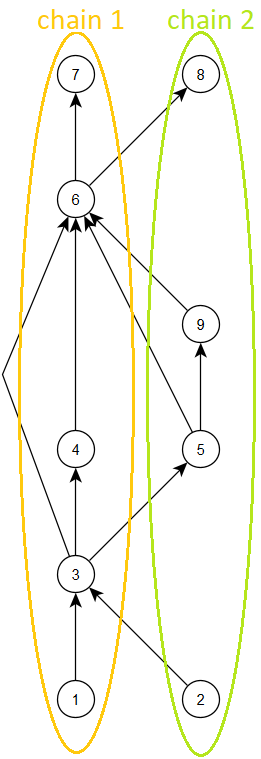}
         \caption{A chain decomposition of a graph consisting of 2 chains.}
     \end{subfigure}
             \caption{Path and chain decomposition of an example graph.}

        \label{fig:Decomposition}
\end{figure}

\subsubsection*{Definitions and Abbreviations}
\begin{itemize}
 \item{\textbf{DAG:}} Directed acyclic graph (DAG) is a directed graph with no directed cycles.
\item \textbf{Path/Chain:} In a path every vertex is connected by a direct edge to its successor, while in a chain any vertex is connected to its successor by a directed path which may be an edge. The vertices of a path/chain are in ascending topological order.
    \item \textbf{Paths/Chains decomposition of a DAG:} Let $G = (V,E)$ be a DAG. A path/chain decomposition of $G$ is a set of vertex-disjoint paths/chains. The decomposition includes all vertices of $G$. There is an example of a path and a chain decomposition in Figure \ref{fig:Decomposition}.
    \begin{itemize}
    \item $k_p$: denotes the number of paths of a path decomposition of a DAG.
    \item $k_c$: denotes the number of chains of a chain decomposition of a DAG.
    \end{itemize}
    \item \textbf{Width:} is the maximum number of mutually unreachable vertices of a graph \cite{DILWORTH}. The number of chains in a minimum chain decomposition of a graph is equal to its width.
     \item \textbf{Transitive edge:} an edge $(v_1,v_2)$ of a DAG $G$ is transitive if there is a path longer than one edge that connects $v_1$ and $v_2$. 
    \item \textbf{In a DAG $G=(V,E)$:}  
    \begin{itemize}
    \item $E_{tr}:$ is the set of all transitive edges of $G$. $E_{tr}\subset E$.
    \item $E_{tr}':$ denotes a subset of $E_{tr}$.
    \item $E_{red}:$ denotes the set of non-transitive edges: $E_{red}=E-E_{tr}$ , $E_{red}\subseteq E$. 
    \item $G_{red} = (V,E_{red}):$ denotes the transitive reduction \cite{TranReduction} of $G = (V,E)$. The transitive reduction is unique for DAGs. It contains the minimum number of edges needed to form the same transitive closure as $G=(V, E)$.
    \end{itemize}
\end{itemize}

\section{DAG Decomposition into Paths}  \label{sec:PathDec}

Jagadish in \cite{Jagadish:1990:CTM:99935.99944} categorized path decomposition techniques into two categories: Chain Order Heuristics and Node Order Heuristics.
The first constructs the paths one by one, while the second creates the paths in parallel. 
He also presented chain decomposition heuristics based on Chain Order and Node Order Heuristic,
utilizing a list of \emph{all successors} and not only the immediate for each vertex.
His algorithms require $O(n^2)$ time using a precomputed transitive closure.  He  states in~\cite{Jagadish:1990:CTM:99935.99944} that in order to build a chain decomposition  we must compute the transitive closure, implying that finding the minimum number of chains has the same complexity as the heuristics. We believe that, in part, this is an important reason for which most of the research on this topic is focusing on finding the optimum number of chains. In this work we argue the opposite:  Namely, we show that we can build an efficient chain decomposition fast and we use it in order to solve the transitive closure problem.  Theoretically, our chain decomposition heuristic has worst-case time complexity $O(n^2)$, without requiring any precomputation of the transitive closure. In practice however it actually behaves like a linear time algorithm as we explain later, both theoretically and experimentally.  Furthermore, in Sections~\ref{sec:HierTrans} and ~\ref{sec:IndexingScheme}, we show how an efficient and practical chain decomposition can crucially enhance transitive closure solutions.

More precisely our algorithm decomposes the graph into $k_c$ chains in $O(|E| + c *l)$ time, where $c$ is the number of path concatenations, and $l$ is the length of a longest path of the graph. An important observation is that the factor $c*l$ depends on the number of successful concatenations c. Hence, the more time our algorithm takes, the more concatenations it completes. If there are no concatenations then $c*l$ equals zero.
In our experiments, the factor $(c*l)$ is less than $|E|$ in almost all cases. We will describe our technique in detail in the next section.  

In the rest of this section, we describe the linear time algorithms for path decomposition using only the immediate successors. Since paths are by definition chains of a special structure, we use the more general term chain in our algorithms.  Clearly, similar algorithms can be used also for chain decomposition. We assume that topological sorting has been performed and examine the vertices in ascending order.

\subsection*{Chain Order Heuristic}
The chain-order heuristic starts from a vertex and keeps on extending the path to the extent possible. The path ends when no more unused immediate successors can be found. The first for loop of~Algorithm~\ref{alg:CO} finds an unused vertex and creates a path. The inner while loop extends the path.

\begin{algorithm}[!ht]
\caption{Path Decomposition CO}\label{alg:CO}
\begin{algorithmic}[1]
\Procedure{ChainOrderHeuristic}{$G,T$}\newline
 \textbf{INPUT:} A DAG $G=(V,E)$, and a topological sorting $T(v_1,...,v_i,...,v_N)$ of G \newline
 \textbf{OUTPUT:} A path decomposition of G
 \State $K \gets \emptyset$ \Comment{Set of paths}
 \State Mark all nodes \textbf{unused}
 
\For{ \text{every }\textbf{unused } \text{vertex }$v_i \in T$ \text{in ascending topological order   }}
\State $current \gets v_i$
\State  $C \gets \mbox{new Chain}()$
\State \textbf{Add} $current$ \textbf{to} $C$
\While{there is an \textbf{unused} immediate successor $s$ of the \textbf{current} node}
  \State \textbf{add} $s$ \textbf{to} $C$
  \State $current \gets s$
\EndWhile
\State \textbf{add} $C$ \textbf{to} $K$
\EndFor
\State \textbf{return} $K$  
\EndProcedure
\end{algorithmic}
\end{algorithm}

 \subsection*{Node Order Heuristic}
 The node-order heuristic examines each vertex (node) $v_i$ and assigns it to an existing chain. If there is no such chain, then a new chain is created for vertex $v_i$. Algorithm~\ref{alg:NO} illustrates the node order heuristic. 
 \begin{algorithm}[!ht]
\caption{Path Decomposition NO}\label{alg:NO}
\begin{algorithmic}[1]
\Procedure{NodeOrderHeuristic}{$G,T$}\newline
 \textbf{INPUT:} A DAG $G=(V,E)$, and a topological sorting $T(v_1,...,v_i,...,v_N)$ of G \newline
 \textbf{OUTPUT:} A path decomposition of G
 \State $K \gets \emptyset$ \Comment{Set of paths}
\For{ \textbf{every vertex }$v_i \in T$ \textbf{in ascending topological order}}
\If{$v_i\mbox{ is an immediate successor of the last node of a chain C}$}
\State \textbf{add} $v_i$ \textbf{to} $C$
\Else
\State  $C \gets \mbox{new Chain}()$
\State \textbf{add} $v_i$ \textbf{to} $C$ 
\State \textbf{add} $C$ \textbf{to} $K$
\EndIf
\EndFor
\State \textbf{return} $K$
\EndProcedure
\end{algorithmic}
\end{algorithm}

\ \

\section{DAG Decomposition into Chains}  \label{sec:ChainDec}
In this section, we present a path concatenation technique that takes as input any path decomposition of a DAG $G=(V,E)$ and constructs a chain decomposition by performing repeated path concatenations. If it performs $c$ path concatenations and $l$ is the length of a longest path of the graph, then it requires $O(|E|+c*l)$ time. In fact, the actual cost for every successful concatenation is the path between the two concatenated paths, which is bounded by $(c*l)$. 
In order to apply our path concatenation algorithm, we first compute a path decomposition of the graph. We propose to use linear-time algorithms based on Node-Order Heuristic or Chain Order Heuristic, as described in the previous section.

\subsection{Path/Chain Concatenation}
Two chains $c_1$ and $c_2$ can merged into a new chain if the last vertex of $c_1$ can reach the first vertex of $c_2$, or if the last vertex of $c_2$ can reach the first vertex of $c_1$. The process of merging two or more chains into a new chain is called path or chain concatenation. To reduce the number of chains of any given chain decomposition, we can find possible concatenations and merge the chains. Searching for a concatenation implies that we are searching for a path between two chains. 
We can start searching from the first vertex of a chain looking for the last vertex of another chain, or from the last vertex of a chain looking for the first vertex of another chain.

Given a DAG $G=(V, E)$ and a path decomposition $D_p$ that contains $k_p$ paths we will build a chain decomposition of $G$ that contains $k_c$ chains in $O(|E|+(k_p-k_c)*l)$ time, where $l$ is the length of a longest path of $G$. This is accomplished by performing path/chain concatenations. Since each path/chain concatenation reduces the number of chains by one, the total number of such concatenations is $(k_p-k_c)$. Since paths are by definition chains of a special structure, and we start by concatenating paths into chains, we use the more general term chain in our algorithms.  


For every path in $D_p$ we start a reversed DFS lookup function from the first vertex of a chain, looking for the last vertex of another chain traversing the edges backward. The reversed DFS lookup function is the well-known depth-first search graph traversal for path finding. If the reversed DFS lookup function detects the last vertex of a chain, then it concatenates the chains. If we simply perform the above, as described, then the algorithm will run in $O(k_p*|E|)$ since we will run $k_p$ reversed DFS searches. However, every reversed DFS search can take advantage of the previous reversed DFS results. A reversed DFS for path finding returns the path between the source vertex and the target vertex, which in our case, is the path between the first vertex of a chain and the last vertex of another chain. Hence, every execution that goes through a set of vertices $V_i$ it splits them into two vertex disjoint sets, $R_i$ and $P_i$.  $P_i$ contains the vertices of the path from the source vertex to the destination vertex and $R_i$ contains every vertex in $V_i-P_i$. If no path is found then $V_i=R_i$ and $P_i=\emptyset$.

Notice that every vertex in the set $R_i$ is not the last vertex of a chain. If it were then it would belong to $P_i$ and not to $R_i$. Similarly, for every vertex in $R_i$, all its predecessors are in $R_i$ too. Hence, if a forthcoming reversed DFS lookup function meets a vertex of $R_i$, there is no reason to proceed with its predecessors.  


\begin{algorithm}
\caption{Concatenation}\label{alg:conc}
\begin{algorithmic}[1]
\Procedure{Concatenation}{$G,D$}\newline
 \textbf{INPUT:} A DAG $G=(V,E)$, and a path decomposition D of G \newline
 \textbf{OUTPUT:} A chain decomposition of G
\For{\textbf{each path:} $p_i \in D$}
\State $f_i \gets \mbox{ first vertex of } p_i$
\State $(R_i,P_i) \gets \mbox{reversed\_DFS\_lookup}(G,f_i)$
\If{$P_i\neq \emptyset$}

\State $l_i \gets \mbox{ destination vertex of } P_i$
\Comment{Last vertex of a path}
\State \textbf{Concatenate\_Paths(} $l_i$, $f_i$ \textbf{)} 
\EndIf
\State $G \gets G \setminus R_i$
\EndFor
\EndProcedure
\end{algorithmic}
\end{algorithm}

Algorithm \ref{alg:conc} shows our path/chain concatenation technique. Observe that the reversed DFS lookup function is invoked for every starting vertex of a path.
Every reversed DFS lookup function goes through the set $R_i$ and the set $P_i$, examining the nodes and their incident edges. $P_i$ is the path from the first vertex of a chain to the last vertex of another chain. The set $R_i$ contains all of the vertices the function went through except the vertices of $P_i$.  Hence we have the following theorem:

\begin{theorem}
The time complexity of Algorithm \ref{alg:conc} is $O(|E|+(k_p-k_c)*l)$, where $l$ is the length of a longest path in $G$.
\end{theorem}

\begin{proof}
Assume that we have $k_p$ paths. We call $k_p$ times the reversed\_DFS\_lookup function. Hence, we have $(R_i, P_i)$ sets, $0\leq i <k_p$.
In every loop, we delete the vertices of $R_i$. Hence, $R_i\cap R_j=\emptyset$ ,$0\leq i,j<k_p$ and $i\neq j$. 
We conclude that $\bigcup\limits_{i=0}^{k_p-1} R_i \subseteq V$ and $\sum_{i=0}^{k_p-1} |R_i| \leq |V|$.

A path/chain $P_i$, $0\leq i < k_p$, is not empty if and only if a concatenation has occurred. Hence, $\sum_{i=0}^{k_p-1}|P_i| \leq c*l$ where $c$ is the number of concatenations and $l$ is the longest path of the graph. 
Since every concatenation reduces the number of paths/chains by one, we have that $c = k_p-k_c$.
\end{proof}

Please notice that according to the previous proof the actual time complexity of Algorithm \ref{alg:conc} is $O(|E|+ \sum_{i=0}^{k_p-1}|P_i|)$ which in practice is expected to be significantly better than $O(|E|+c*l)$. 
Indeed, this is confirmed by our experimental results. We ran some extra experiments on ER
graphs of size 10k, 20k, 40k, 80k, and 160k vertices and average degree 10; the run times were 9, 34, 99,
228, and 538 milliseconds, respectively, which shows again that the execution time is almost linear.

\subsection{A Better Chain Decomposition Heuristic} \label{section H3}

Algorithm \ref{alg:conc} describes how to produce a chain decomposition of a DAG $G$ by applying a path/chain concatenation technique. Notice that this algorithm works for any given path/chain decomposition of $G$.  Of course, the number of future potential concatenations depends on the starting path/chain decomposition of $G$, that is given as input and it will play a crucial role in obtaining a solution that is close to optimal. Hence, we will present an improved heuristic, called Algorithm~\ref{alg:H3}, whose goal is to produce a path decomposition  of $G$ that will allow a large number of future concatenations.  Since these variations are not computationally intensive, Algorithm~\ref{alg:H3} runs in linear time.  Next, using the output path decomposition produced by Algorithm~\ref{alg:H3} we will describe Algorithm~\ref{alg:H3_conc} that computes a chain decomposition that is very close to the optimal, according to our experimental results, and still runs in $O(|E|+c*l)$ time.

Algorithm~\ref{alg:H3}, which is a variation of the Node Order Heuristic (Algorithm \ref{alg:NO}), follows the same philosophy but with two important additions that aim to construct a path decomposition in which more concatenations will be possible in the future: (a) when we visit a vertex of out-degree 1, we immediately add its unique immediate successor to its path/chain, and (b) instead of searching for the first available immediate predecessor (that is the last vertex of a path), we choose an available vertex with the lowest out-degree. This heuristic aims to create a path decomposition that will allow more concatenations in the future. Since Algorithm~\ref{alg:H3} goes through all vertices of $G$, and for every vertex, it examines all the outgoing (line 8) and all the incoming edges (line 19), its time complexity is linear.

We are ready now to present our best chain decomposition algorithm, ~Algorithm~\ref{alg:H3_conc}, which is a combination of~Algorithm~\ref{alg:H3} and the chain concatenation described in~Algorithm~\ref{alg:conc}.  The only addition to~Algorithm~\ref{alg:H3} is the block of the if-statement of line 10. In other words, if we do not find an immediate predecessor, we search all predecessors using the reversed\_DFS\_lookup function. The differentiation from our concatenation approach is that it does not use it as a post-processing step. It is applied in real time if the algorithm does not find an immediate predecessor that is the last vertex of a chain. This is done in order to avoid transitive edges that could lead to false matches.

 \begin{algorithm}[hbt!]
\caption{Path Decomposition (H3)}\label{alg:H3}
\begin{algorithmic}[1]
\Procedure{Node-Order based variation}{$G,T$}\newline
 \textbf{INPUT:} A DAG $G=(V,E)$, and a topological sorting $T(v_1,...,v_i,...,v_N)$ of G \newline
 \textbf{OUTPUT:} A path decomposition of G
 \State $K \gets \emptyset$ \Comment{Set of paths}
\For{ \textbf{every vertex }$v_i \in T$ \textbf{in ascending topological order}}
\State $\mbox{Chain }C$ \Comment{C is a pointer to a path}
\If{$u_i \mbox{ is assigned to a chain}$}
\State $C \gets v_i\mbox{'s chain}$
\ElsIf{$v_i\mbox{ is not assigned to a chain}$}
\State $l_i\gets\mbox{choose the immediate predecessor with the lowest outdegree }$
\State $\mbox{\hspace*{0.8cm}that is the last vertex of a chain}$
\If{$l_i\neq \mbox{null}$}
\State $C \gets \mbox{path indicated by }l_i$
\State \textbf{add} $v_i$ \textbf{to} $C$
\Else
\State  $C \gets \mbox{new Chain}()$
\State \textbf{add} $v_i$ \textbf{to} $C$ 
\EndIf
\State \textbf{add} $C$ \textbf{to} $K$
\EndIf
\If{there is an immediate successor $s_i$ of $v_i$ with in-degree 1}
\State \textbf{add} $s_i$ \textbf{ to } $C$ 
\EndIf

\EndFor
\State \textbf{return} $K$
\EndProcedure
\end{algorithmic}
\end{algorithm}

\begin{algorithm}[hbt!]
\caption{Chain Decomposition (H3 conc.)}\label{alg:H3_conc}
\begin{algorithmic}[1]
\Procedure{NodeOrder based variation with concatenation}{$G,T$}\newline
 \textbf{INPUT:} A DAG $G=(V,E)$, and a topological sorting $T(v_1,...,v_i,...,v_N)$ of G \newline
 \textbf{OUTPUT:} A chain decomposition of G
 \State $K \gets \emptyset$ \Comment{Set of chains}
\For{ \textbf{every vertex }$v_i \in T$ \textbf{in ascending topological order}}
\State $\mbox{Chain }C$
\If{$v_i \mbox{ is assigned to a chain}$}
\State $C \gets v_i\mbox{'s chain}$  \Comment{C is a pointer a to path}
\ElsIf{$v_i\mbox{ is not assigned to a chain}$}
\State $l_i\gets\mbox{choose the immediate predecessor with the lowest outdegree }$
\State $\mbox{\hspace*{0.8cm}that is the last vertex of a chain}$
\If{$l_i = \mbox{null}$}
\State $(R_i,P_i) \gets \mbox{reverse\_DFS\_lookup(}G,v_i\mbox{)}$
\If{$P_i\neq \emptyset$}
\State $l_i \gets \mbox{ destination vertex of } P_i$
\EndIf
\State $G \gets G \setminus R_i$
\EndIf
\If{$l_i\neq \mbox{null}$}
\State $C \gets \mbox{chain indicated by }l_i$
\State \textbf{add} $v_i$ \textbf{to} $C$
\Else
\State  $C \gets \mbox{new Chain}()$
\State \textbf{add} $v_i$ \textbf{to} $C$ 
\EndIf
\State \textbf{add} $C$ \textbf{to} $K$
\EndIf
\If{there is an immediate successor $s_i$ of $v_i$ with in-degree 1}
\State \textbf{add} $s_i$ \textbf{ to } $C$ 
\EndIf

\EndFor
\State \textbf{return} $K$
\EndProcedure
\end{algorithmic}
\end{algorithm}

\subsection{DAG Decomposition: Experimental Results}
\label{ExperimentsChains}

In this section we present extensive experimental results on graphs, most of which were created by NetworkX~\cite{networkx}. We use three different random graph generator models: Erdős-Rényi~\cite{erdHos1959renyi}, Barabasi-Albert~\cite{barabasi1999emergence}, and Watts-Strogatz~\cite{watts1998collective} models.
The generated graphs are made acyclic, by orienting all edges from low to high ID number, see the documentation of networkx~\cite{networkx} for more information about the generators.  Additionally, we use the Path‑Based DAG Model that was introduced in~\cite{lionakis2021constant} and is especially designed for DAGs with a predefined number of randomly created paths. For every model, we created 12 graphs: Six graphs of 5000 nodes and six graphs of 10000 nodes and average degree 5, 10, 20, 40, 80, and 160. We ran the heuristics on multiple copies of the graphs and examined the performance of the heuristics in terms of the number of chains in the produced chain decompositions. Additionally, we observed that the graphs generated by the same generator with the same parameters have small width deviation. For example, the percentage of deviation on ER and Path-Based model is about 5\% and for the BA model is less than 10\%. The width deviation of graph in the WS model  is a bit higher, but this is expected since the width of these graphs is smaller.
\newline
\newline
\textbf{Random Graph Generators:}
\begin{itemize}
  \item \textbf{Erdős-Rényi (ER) model} \cite{erdHos1959renyi}: The generator returns a random graph $G_{n,p}$, where $n$ is the number of nodes and every edge is formed with probability $p$.
  \item \textbf{Barabási–Albert (BA) model} \cite{barabasi1999emergence}: preferential attachment model: A graph of $n$ nodes is grown by attaching new nodes each with $m$ edges that are preferentially attached to existing nodes with high degree. The factors $n$ and $m$ are parameters to the generator.
  \item \textbf{Watts–Strogatz (WS) model} \cite{watts1998collective}: small-world graphs: First it creates a ring over $n$ nodes. Then each node in the ring is joined to its $k$ nearest neighbors. Then shortcuts are created by replacing some edges as follows: for each edge $(u,v)$ in the underlying “$n$-ring with $k$ nearest neighbors” with probability $b$ replace it with a new edge $(u,w)$ with uniformly random choice of an existing node $w$. The factors $n,k,b$ are the parameters of the generator.
  \item \textbf{Path‑Based DAG (PB) model} \cite{lionakis2021constant}:  In this model, graphs are randomly generated based on a number of predefined but randomly created paths.
\end{itemize}


We compute the minimum set of chains using the method of Fulkerson~\cite{Fulkerson}, which in brief consists of the following steps:  1) Construct the transitive closure $G^*(V,E')$ of $G$, where $V=\{v_1,...,v_n\}$. 2) Construct a bipartite graph $B$ with partitions $(V_1, V_2)$, where $V1=\{x_1, x_2,..., x_n\},$ $ V2=\{y_1, y_2, ..., y_n\} $. An edge $(x_i,y_j)$ is formed whenever $(v_i,v_j)\in E'$.  3) Find a maximal matching $M$ of $B$. The width of the graph is $n-|M|$. In order to construct the minimum set of chains, for any two edges $e_1,e_2\in M$, if $e_1=(x_i,y_t)$ and $e_2=(x_t,y_j)$ then connect $e_1$ to $e_2$.

The aim of our experiments is twofold: (a) to understand the behavior of the width of DAGs created in different models, and (b) to compare the behavior of all the heuristics used on graphs of these models.
Table \ref{fig:chains5000} shows the width and the  number of chains created by the various heuristics for every graph of 5000 nodes. Table \ref{fig:chains10000} shows the same for graphs of 10000 nodes. 
\begin{itemize}
  \item \textbf{CO}: Path decomposition using Chain Order Heuristic, (Algorithm \ref{alg:CO})
  \item \textbf{CO conc.}: Chain decomposition using Chain Order Heuristic and our concatenation technique. (Algorithm \ref{alg:CO} followed by Algorithm \ref{alg:conc})
  \item \textbf{NO}: Path decomposition using Node Order Heuristic. (Algorithm \ref{alg:NO})
  \item \textbf{NO conc.}: Chain decomposition using Node Order Heuristic and our concatenation technique. (Algorithm \ref{alg:NO} followed by Algorithm \ref{alg:conc})
  \item \textbf{H3}: Path decomposition using Algorithm \ref{alg:H3}, our NO Heuristic variation 
  \item \textbf{H3 conc.}: Chain decomposition using Algorithm \ref{alg:H3_conc}
  \item \textbf{Width}: The width of the graph (computed by Fulkerson's method). 
\end{itemize}

Observe that in both tables our chain decomposition heuristic, H3 conc., performs better than all other heuristics since it produces fewer chains. 

In order to understand the behavior of the width on DAGs of these different models we observe: (i) the BA model produces graphs with a larger width than ER, and (ii) the ER model creates graphs with a larger width than WS.  For the WS model, we created two sets of graphs: The first has probability $b$ equals 0.9 and the second 0.3. Clearly, if the probability $b$ of rewiring an edge is 0, the width would be one, since the generator initially creates a path that goes through all vertices. As the rewiring probability $b$ grows, the width grows. That is the reason we choose a low and a high probability. Figure \ref{fig:width5000} and \ref{fig:width10000} demonstrates the behavior of the width for each model on the graphs of 5000 and 10000 nodes. Another interesting observation is that the width of the ER model follows the curve $width =\frac{\mbox{nodes}}{\mbox{average degree}}$. 

In order to compare the number of chains produced by all heuristics used on the graphs of these models we created several figures.  Namely, we visualize how close is the result of our heuristics to the width. In Figures \ref{fig:Chart_WidthBarabasiComp10000}, \ref{fig:Chart_WidthErdosComp10000}, and \ref{fig:Chart_WidthWatts09Comp10000}, we show how close is the number of chains produced by our technique (i.e., the blue line) to the width (i.e., red line) for ER, BA, and WS models. A comparison between the ER and the Path-Based models is shown in Figure~\ref{fig:Chart_ErdosPBComp10000} for graphs of 10000 nodes and varying average degree. It is interesting to note that for sparser PB graphs the width is very close to the number of predefined (but randomly created paths) whereas the width of the ER graphs is very high.  However, when the graphs become denser the width for both models seem to eventually converge.



All heuristics run very fast, in just a few milliseconds. Thus it is not interesting to elaborate on their running time. 
In the following sections, we present partial run-time results that are obtained for computing an indexing scheme (see Tables~\ref{tables:CO_H3conc_scheme}, \ref{table:IndexingSchemeResults5000}, and~\ref{table:IndexingSchemeResults10000}).  These results indicate that the running time of the heuristics is indeed very low.

\begin{table}
\centering
\includegraphics[scale=0.7]{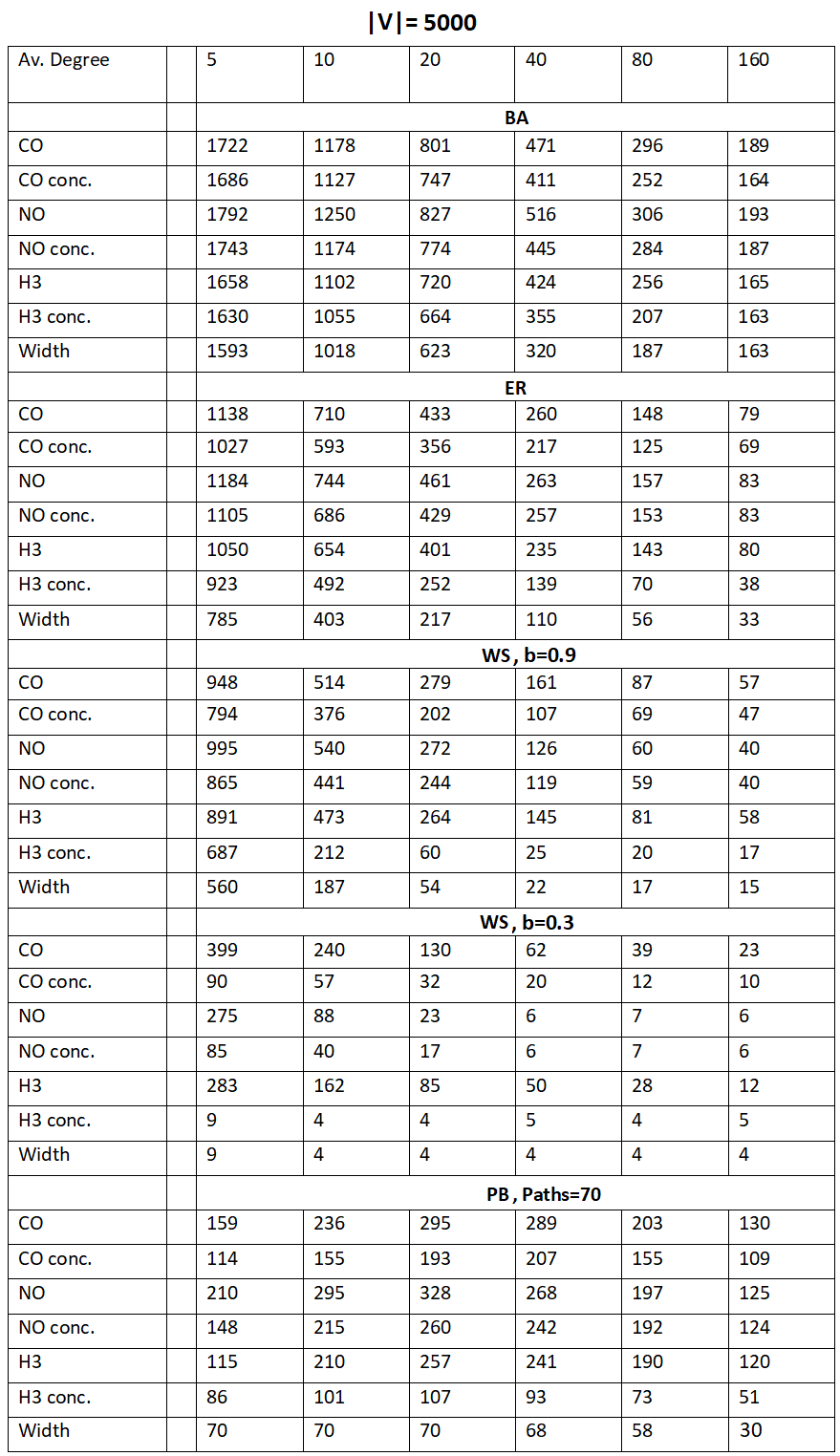}
\caption{Comparing the number of chains produced by the path and chain decomposition algorithms on graphs with 5000 nodes.}
\label{fig:chains5000}
\end{table}

\begin{table}
\centering
\includegraphics[scale=0.7]{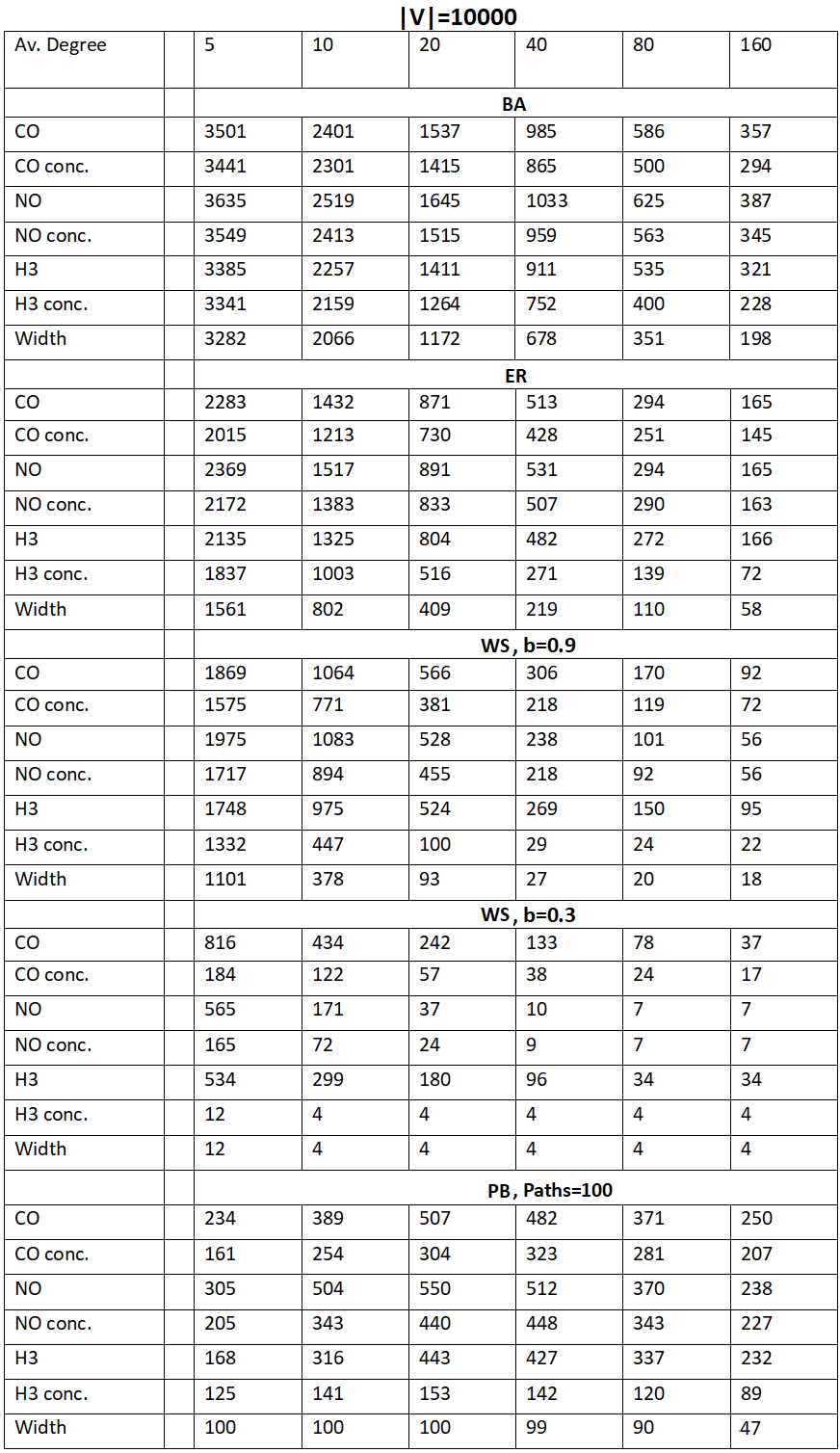}
\caption{Comparing the number of chains produced by the path and chain decomposition algorithms on graphs with 10000 nodes.}
\label{fig:chains10000}
\end{table}
\begin{figure}
     \centering
     \begin{subfigure}[b]{0.8\textwidth}
         \centering
         \includegraphics[width=\textwidth]{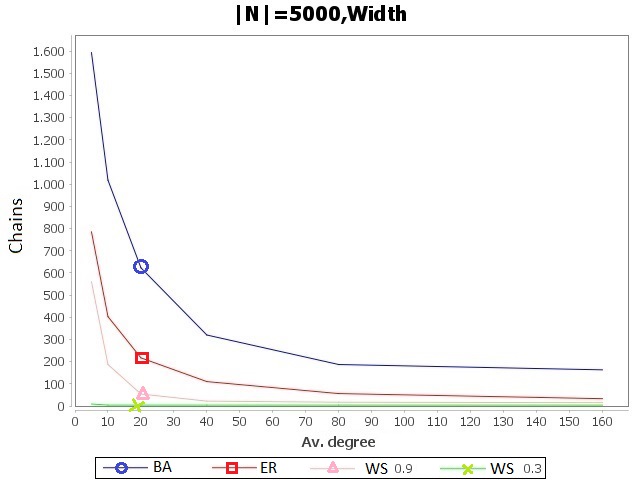}
         \caption{The width curve on graphs of 5000 nodes.}
         \label{fig:width5000}
     \end{subfigure}
     \hfill
     \begin{subfigure}[b]{0.8\textwidth}
         \centering
         \includegraphics[width=\textwidth]{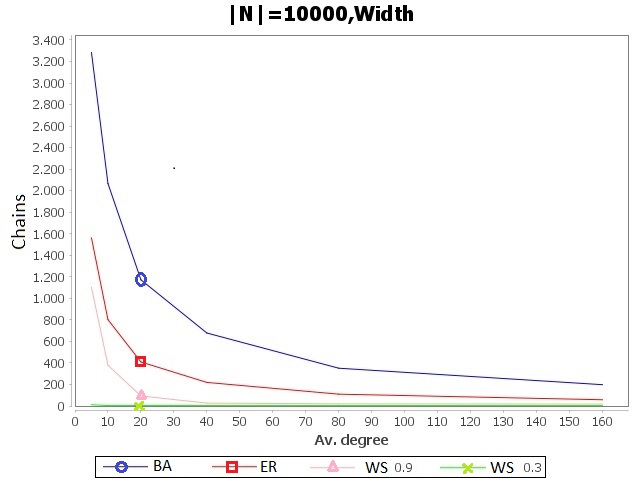}
         \caption{The width curve on graphs of 10000 nodes.}
         \label{fig:width10000}
     \end{subfigure}
    \caption{The width curve on graphs of 5000 and 10000 nodes using three different models.
}
        \label{fig:width_comp}
\end{figure}

\begin{figure}
\centering
 \includegraphics[width=0.9\textwidth]{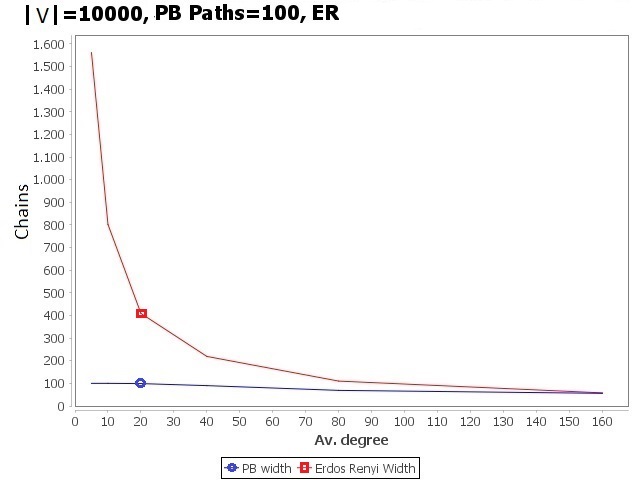}
\caption{A comparison of width between ER model and PB model.}
\label{fig:Chart_ErdosPBComp10000}
\end{figure}

\begin{figure}
\centering
\includegraphics[width=0.9\textwidth]{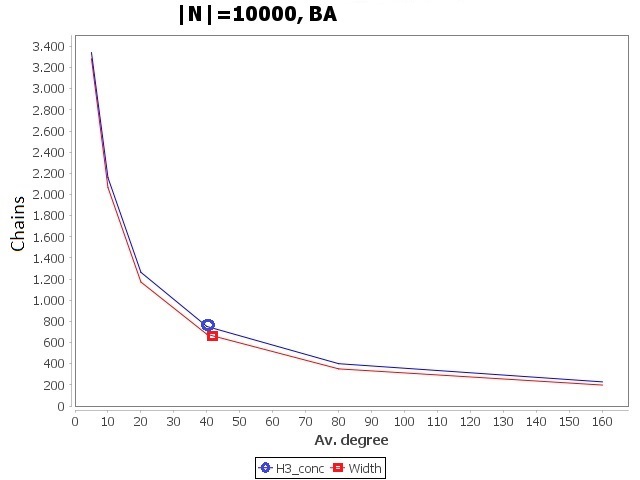}
\caption{The efficiency of the chain decomposition algorithm in the BA model.}
\label{fig:Chart_WidthBarabasiComp10000}
\end{figure}

\begin{figure}
\centering
\includegraphics[width=0.9\textwidth]{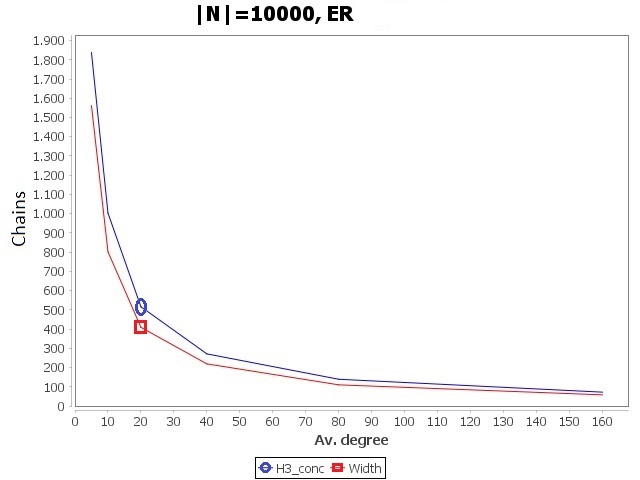}
\caption{The efficiency of our chain decomposition algorithm in the ER model.}
\label{fig:Chart_WidthErdosComp10000}
\end{figure}

\begin{figure}
\centering
 \includegraphics[width=0.9\textwidth]{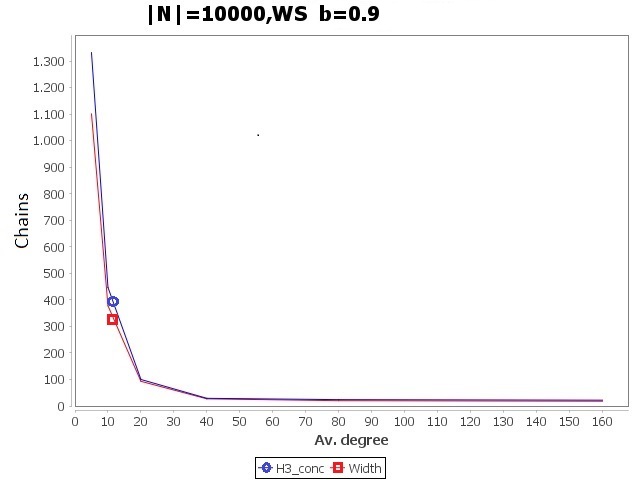}
\caption{The efficiency of our chain decomposition algorithm in WS model.}
\label{fig:Chart_WidthWatts09Comp10000}
\end{figure}

\begin{figure}
\centering
 \includegraphics[width=0.9\textwidth]{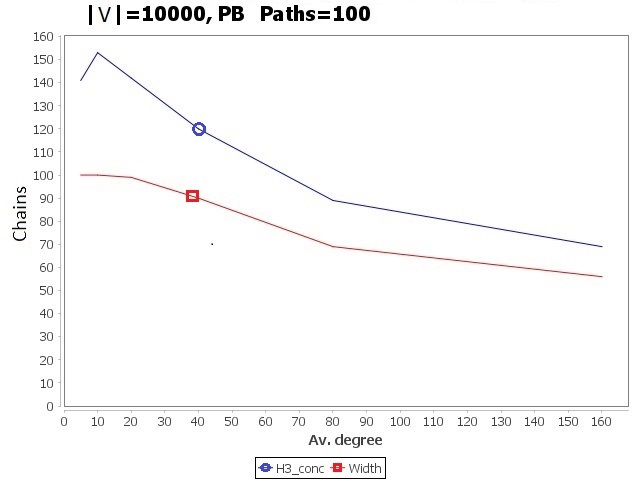}
\caption{The efficiency of our chain decomposition algorithm in the PB model.}
\label{fig:Chart_WidthPBMComp10000}
\end{figure}
\clearpage
\section{Hierarchies and Faster Transitivity}\label{sec:HierTrans}
The importance of removing transitive edges in order to create an abstract graph utilizing paths and chains was first described in~\cite{lionakis2019adventures}.  Their focus was on graph visualization techniques, while in this work we apply their abstraction to the transitive closure problem.  Hence we state the following useful observations:
\begin{enumerate}
    \item Given a chain decomposition $D$ of a DAG $G=(V,E)$, each vertex $v_i \in V$, $0\leq i < |V|$, can have at most one outgoing non-transitive edge directed to a vertex of every other chain.
    \item Given a chain decomposition $D$ of a DAG $G=(V,E)$, each vertex $v_i \in V$, $0\leq i < |V|$, can have at most one incoming non-transitive edge directed from a vertex of every other chain.
    \item Let $G=(V,E)$ be a DAG with width $w$. The non-transitive edges of $G$ are less than or equal to $width*|V|$, in other words $|E_{red}|=|E|-|E_{tr}|\leq width*|V|$.
\end{enumerate}



\begin{figure}[ht!]
     \centering
     \begin{subfigure}[b]{0.25\textwidth}
         \centering
         \includegraphics[width=\textwidth]{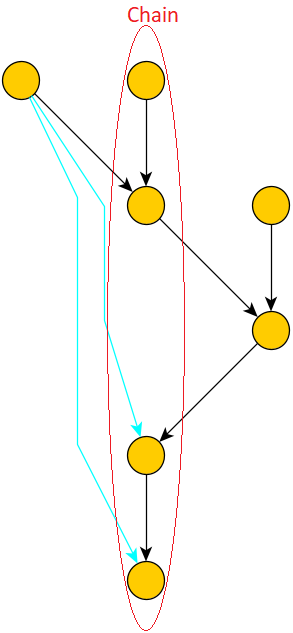}
         \caption{}
         \label{fig:transEdges_a}
     \end{subfigure}
     \hfill
     \begin{subfigure}[b]{0.25\textwidth}
         \centering
         \includegraphics[width=\textwidth]{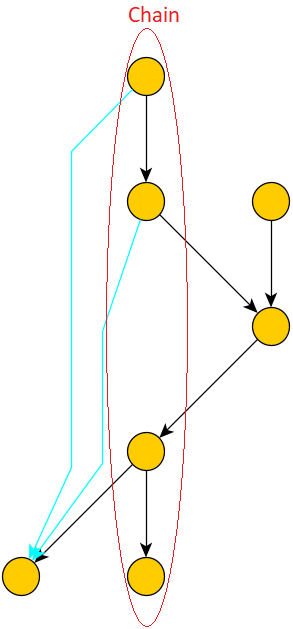}
         \caption{}
         \label{fig:transEdges_b}
     \end{subfigure}
        \caption{
        The light blue edges are transitive. (a) shows the outgoing transitive edges that end in the same chain. (b) shows the incoming transitive edges that start from the same chain.}
        \label{fig:transEdges}
\end{figure}


An interesting application of the above observations
is that we can find a significantly large subset of $E_{tr}$ in linear time as follows: Given any chain (or path) decomposition with $k_c$ chains, we can trace the vertices and their outgoing edges and keep the edges that point to the lowest point of each chain, rejecting the rest as transitive. We do the same for the incoming edges keeping the edges that come from the highest point (i.e., the vertex with the highest topological rank) of each chain. This way, we can find a significantly large subset $E_{tr}' \subseteq E_{tr}$. Hence, $|E-E_{tr}'|\leq k_c*|V|$. Clearly, this approach can be used as a linear-time preprocessing step in order to reduce the size of any DAG. Consequently, this will speed up every transitive closure algorithm bounding the number of edges of an input graph, and the indegree and outdegree of every vertex by $k_c$.
For example, algorithms based on tree cover, see \cite{agrawal1989efficient,chen2005stack,trissl2007fast,wang2006dual},
are practical on sparse graphs and can be enhanced further with such a preprocessing step that removes transitive edges. Additionally, this approach may have practical applications in dynamic transitive closure techniques: If we answer queries on time using graph traversal for every query, we could reduce the size of the graph with a fast (linear-time) preprocessing step that utilizes chains. Also, we could quickly decide if the edges we add and/or remove are transitive (or not). Transitive edges do not affect the transitive closure, hence no updates are required. This could be practically useful if we are adding/removing edges in a  dynamic setting.

\section{Reachability Indexing Scheme}  \label{sec:IndexingScheme}
In this section, we present an important application of a fast and practical chain decomposition technique. Namely, we solve the transitive closure problem by creating a reachability indexing scheme that is based on chain decomposition and we evaluate it by running extensive experiments. Our experiments shed light on the interplay of various important factors as the density of the graphs increases.

Jagadish described a compressed transitive closure technique in 1990~\cite{Jagadish:1990:CTM:99935.99944} applying the indexing scheme and path/chain decomposition. His method uses successor lists and focuses on the compression of the transitive closure.  Thus his scheme does not answer queries in constant time. As discussed above, Jagadish's heuristic for chain decomposition runs in $O(n^2)$ time using a pre-computed transitive closure. Our technique runs in almost linear time without requiring a pre-computed transitive closure, and the result is close to the optimal. 
Simon~\cite{simon}, describes a technique similar to~\cite{Jagadish:1990:CTM:99935.99944}. His technique is based on computing a path decomposition, thus boosting the method presented in~\cite{goralcikova1979}. The linear time heuristic used by Simon is very similar to the Chain Order Heuristic  of~\cite{Jagadish:1990:CTM:99935.99944}, also described earlier. 
A different example is a graph structure referred to as path-tree cover introduced in ~\cite{jin2008efficiently}, where the authors utilize a path decomposition algorithm to build their labeling.

In the following subsections, we describe how to compute an indexing scheme in $O(k_c*|E_{red}|)$ time, where $k_c$ is the number of chains (in any given chain decomposition) and $|E_{red}|$ is the number of non-transitive edges. Following the observations of section~\ref{sec:HierTrans},
the time complexity of the scheme
 can be expressed as $O(k_c*|E_{red}|) = O(k_c*width*|V|)$ since $|E_{red}| \leq width*|V|$.  Our scheme utilizes arrays of indices, in a similar fashion as Simon's technique~\cite{simon}, to answer queries in constant time. The space complexity is $O(k_c*|V|)$.
 
 For our experiments we utilize our chain decomposition approach, which produces smaller decompositions than previous techniques, without any considerable run-time overhead.  Thus the indexing scheme is more efficient both in terms of time and space requirements.
 The experimental work shows that the chains rarely have the same length. Usually, a decomposition consists of a few long chains and several short chains. Hence, for most graphs it is not even possible to have $|E_{red}|=width*|V|$. In fact, $|E_{red}|$ is usually  much lower than that. The experimental results presented in Tables \ref{table:IndexingSchemeResults5000} and \ref{table:IndexingSchemeResults10000} confirm this observation in practice.


Given a directed graph with cycles, we can find the strongly connected components (SCC) in linear time. Since any vertex is reachable from any other vertex in the same SCC (they form an equivalence class), all vertices in a SCC can be collapsed into a supernode. Hence, any reachability query can be reduced to a query in the resulting directed acyclic graph (DAG). This is a well-known step that has been widely used in many applications. Therefore, with loss of generality, we assume that the input graph to our method is a DAG. The following general steps describe how we compute the reachability indexing scheme:
\begin{enumerate}
  \item Compute a Chain decomposition
  \item Sort all Adjacency Lists
  \item Create an Indexing Scheme 
\end{enumerate}
In Step 1, we use our chain decomposition technique that runs in $O(|E|+c*l)$ time.
In Step 2, we sort all the adjacency lists in $O(|V|+|E|)$ time. Finally, we create an indexing scheme in $O(k_c*|E_{red}|)$ time and $O(k_c*|V|)$ space. 
Clearly,  if the algorithm of Step 1 computes fewer chains then Step 3 becomes more efficient in terms of time and space.

\subsection{The Indexing Scheme}
Given any chain decomposition of a DAG $G$ with size $k_c$, an indexing scheme will be computed for every vertex that includes a pair of integers and an array of size $k_c$ of indexes. See for example Figure \ref{IndexingSchemeExample}. The first integer of the pair indicates the node's chain and the second its position in the chain. For example, vertex 1 of Figure \ref{IndexingSchemeExample} has a pair $(1,1)$. This means that vertex 1 belongs to the $1st$ chain, and it is the $1st$ element in it.
Given a chain decomposition, we can easily construct the pairs in $O(|V|)$ time with a traversal of the chains.
Every entry of the $k_c$-size array represents a chain. The $i$-th cell represents the $i$-th chain. The entry in the $i$-th cell corresponds to the lowest point of the $i$-th chain the vertex can reach. For example, the array of vertex 1 is $[1,3,2]$. The first cell of the array indicates that vertex 1 can reach the first vertex of the first chain (can reach itself, reflexive property). The second cell of the array indicates that vertex 1 can reach the third vertex of the second chain (There is a path from vertex 1 to vertex 7). Finally, the third cell of the array indicates that vertex 1 can reach the second vertex of the third chain.

Notice that we do not need the second integer of all pairs. If we know the chain a vertex belongs in, we can conclude its position using the array. We use this presentation to simplify the users' understanding.

The process of answering a reachability query is simple. Assume, there is a vertex $s$ and a target vertex $t$. To find if the vertex $t$ is reachable from the $s$, we get $t$'s chain, and we use it as an index in $s$'s array. Hence, we know the lowest point of $t$'s chain vertex $s$ can reach. $s$ can reach $t$ if that point is less than or equal to $t$'s position, else it cannot.

\begin{figure}
\centering
\includegraphics[width=0.6\textwidth]{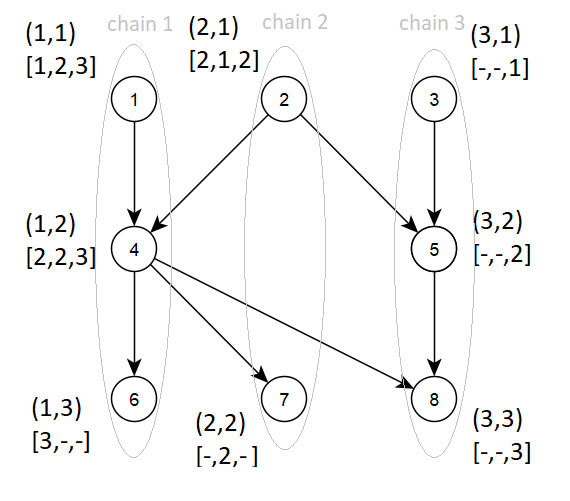}
\caption{An example of an indexing scheme.}
\label{IndexingSchemeExample}
\end{figure}

\subsection{Sorting Adjacency lists}
Next we present an algorithm, Algorithm \ref{alg:SortingAdjList}, that sorts all the adjacency lists of immediate successors in ascending topological order in linear time.  The variable stack indicates the sorted adjacency list. The algorithm traverses the vertices in reverse topological order $(v_n,...,v_1)$. For every vertex $v_i$, $1\leq i\leq n$, it pushes $v_i$ in the stacks of all immediate predecessors.
This step may be performed even before the chain decomposition step, as a preprocessing step. We present it in this section to emphasize its crucial role in the efficient creation of the indexing scheme. If the adjacency lists are not sorted then the time complexity of the algorithm would be $O(k_c*|E|)$ instead of $O(k_c*|E_{red}|)$.

\begin{algorithm}[!ht]
\caption{Sorting Adjacency lists}
\label{alg:SortingAdjList}
\begin{algorithmic}[1]
\Procedure{Sort}{$G,t$}\newline
 \textbf{INPUT:} A DAG $G=(V,E)$ 
\For{\textbf{each vertex:} $v_i \in G$}
\State $v_i\mbox{.stack  } \gets \mbox{new stack()}$
\EndFor
\For{\textbf{each vertex $v_i$ in reverse topological order} }
\For{\textbf{every incoming edge $e(s_j,v_i)$ } }
\State $s_j\mbox{.stack.push($v_i$)} $ 
\EndFor
\EndFor
\EndProcedure
\end{algorithmic}
\end{algorithm}

\begin{lemma}
\label{lemma:SortingAdjList}
Algorithm \ref{alg:SortingAdjList} sorts the adjacency lists of immediate successors in ascending topological order.
\end{lemma}
\begin{proof}
Assume that there is a stack $(u_1, ..., u_n)$, $u_1$ is the top of the stack. Assume that there is a pair $(u_j, u_k)$ in the stack, where $u_j$ has a bigger topological rank than $u_k$ and $u_j$ precedes $u_k$. That means that the for-loop examined $u_j$ before $u_k$ since it goes through the vertices in reverse topological order. This is a contradiction. The vertex $u_j$ cannot precede $u_k$ if it were examined first by the for-loop.
\end{proof}

\subsection{Creating the Indexing Scheme}
Now we present Algorithm \ref{alg:IndexingScheme} that constructs the indexing scheme. The first for-loop initializes the array of indexes. For every vertex, it initializes the cell that corresponds to its chain. The rest of the cells are initialized to infinite. The indexing scheme initialization is illustrated in figure \ref{IndexingSchemeExample_init}. The dashes represent the infinite. Notice that after the initialization, the indexes of all sink vertices have been calculated. Since a sink has no successors, the only vertex it can reach is itself.

\begin{algorithm}
\caption{Indexing Scheme}
\label{alg:IndexingScheme}
\begin{algorithmic}[1]
\Procedure{Create Indexing Scheme}{$G,T,D$}\newline
 \textbf{INPUT:} A DAG $G=(V,E)$, a topological sorting T of G, and the decomposition D of G.
\For{\textbf{each vertex:} $v_i \in G$}
\State $v_i\mbox{.indexes  } \gets \mbox{new table[size of D]}$
\State $v_i\mbox{.indexes.fill(} \infty \mbox{)}$
\State $ch\_no \gets v_i\mbox{'s chain index} $
\State $pos \gets v_i\mbox{'s chain position} $
\State $v_i\mbox{.indexes[ } ch\_no \mbox{ ]} \gets pos$
\EndFor
\For{\textbf{each vertex $v_i$ in reverse topological order} }
\While{$v_i$.stack $\neq \emptyset$ }
\State $target \gets v_j\mbox{.stack.pop()} $ 
\State $t\_ch \gets target\mbox{'s chain index} $
\State $t\_pos \gets target\mbox{'s chain position} $
\If{ $t\_pos<v_i$.indexes[$t\_ch$]}
\Comment{ $(v_i,target)$ is not transitive}
\State $v_i\mbox{.updateIndexes}(target.indexes)$
\EndIf
\EndWhile
\EndFor
\EndProcedure
\end{algorithmic}
\end{algorithm}

\begin{figure}
\centering
\includegraphics[width=0.6\textwidth]{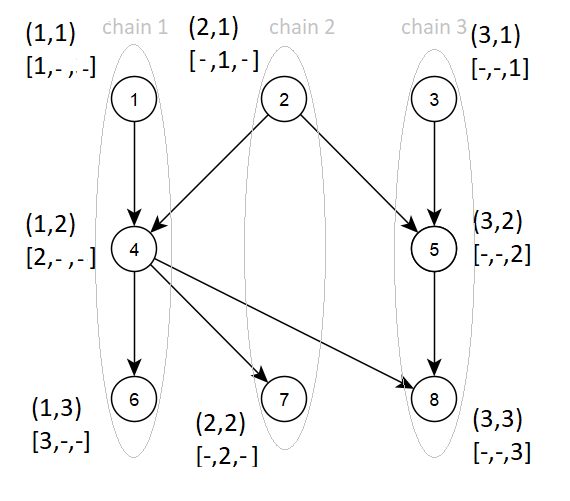}
\caption{Initialization of indexes.}
\label{IndexingSchemeExample_init}
\end{figure}

The second for-loop builds the indexing scheme. It goes through vertices in descending topological order. For each vertex, it visits its immediate successors (outgoing edges) in ascending topological order and updates the indexes.
Suppose we have the edge $(v,s)$, and we have calculated the indexes of vertex $s$ ($s$ is immediate successor of $v$). The process of updating the indexes of $v$ with its immediate successor $s$ means that s will pass all its information to the vertex $v$. Hence, vertex $v$ will be aware that it can reach $s$ and all its successors. Assume the array of indexes of $v$ is $[a_1,a_2,...,a_{k_c}]$ and the array of $s$ is $[b_1,b_2,...,b_{k_c}]$. To update the indexes of $v$ using $s$, we merely trace the arrays and keep the smallest values. For every pair of indexes $(a_i, b_i)$, $0\leq i < kc$, the new value of $a_i$ will be min\{$a_i$ , $b_i$\}. This process needs $k_c$ steps.

\begin{lemma}
\label{lemma:calc_indexes}
Given a vertex v and the calculated indexes of its successors, the while-loop of algorithm \ref{alg:IndexingScheme} (lines 10-17) calculates the indexes of v by updating its array with its non-transitive outgoing edges' successors.
\end{lemma}
\begin{proof}
Updating the indexes of vertex $v$ with all its immediate successors will make $v$ aware of all its descendants. The while-loop of Algorithm \ref{alg:IndexingScheme} does not perform the update function for every direct successor. It skips all the transitive edges. Assume there is such a descendant $t$ and the transitive edge $(v,t)$. Since the edge is transitive, we know by definition that there exists a path from $v$ to $t$ with a length of more than 1. Suppose that the path is $(v,v_1,..,t)$. Vertex $v_1$ is a predecessor of $t$ and immediate successor of $v$. Hence it has a lower topological rank than $t$. Since, while-loop examines the incident vertices in ascending topological order, then vertex $t$ will be visited after vertex $v_1$. The opposite leads to a contradiction.
Consequently, for every incident transitive edge of $v$, the loop firstly visits a vertex $v_1$ which is a predecessor of $t$. Thus vertex $v$ will be firstly updated by $v_1$ and it will record the edge $(v,t)$ as transitive. There is no reason to update vertex $v$ indexes with those of vertex $t$ since the indexes of $t$ will be greater or equal.
\end{proof}
\begin{theorem}
\label{thm:IndexingSchemeTime}
Let $G=(V,E)$  be a DAG. Algorithm \ref{alg:IndexingScheme} computes an indexing scheme for $G$ in $O(k_c*|E_{red}|+|E_{tr}|)$ time.
\end{theorem}
\begin{proof}
In the initialization step, the indexes of all sink vertices have been computed as we described above. 
Taking vertices in reverse topological order, the first vertex we meet is a sink vertex. When the for-loop of line $9$ visits the first non-sink vertex, the indexes of its successors are computed (all its successors are sink vertices).
According to Lemma 1, we can calculate its indexes, ignoring the transitive edges. Assume the for-loop has reached vertex $v_i$ in the $i$th iteration, and the indexes of its successors are calculated. 
Following Lemma 1, we can calculate its indexes. Hence, by induction, we can calculate the indexes of all vertices, ignoring all $|E_{tr}|$ transitive edges in $O(k_c*|E_{red}|+|E_{tr}|)$ time.
\end{proof}

As described in the introduction, a parameterized linear-time algorithm for computing the minimum number of chains was recently presented in~\cite{caceres2022sparsifying}.
Its time complexity is $O(k^3|V | + |E|)$ 
where $k$ is the minimum number of chains, which is equal to the width of $G$.
Using their result Algorithm \ref{alg:IndexingScheme} computes an indexing scheme for $G$ in parameterized linear time. This implies that the transitive closure of $G$ can be computed in parameterized linear time. Hence we have the following:

\begin{corollary}
    \label{cor:IndexingSchemeTime}
Let $G=(V,E)$ be a DAG. Algorithm \ref{alg:IndexingScheme} computes an indexing scheme for $G$ in parameterized linear time.
Hence the transitive closure of $G$ can be computed in parameterized linear time.
\end{corollary}

\subsection{Experimental Results}
We conducted an extensive experimental evaluation of our techniques and we report the results here.  
Tables \ref{tables:CO_H3conc_scheme_a} and \ref{tables:CO_H3conc_scheme_b} include the number of chains and real running time in ms, for the computation of an indexing scheme using two indicative decomposition techniques on ER graphs of 10000 nodes with varying average degree.
In Table \ref{tables:CO_H3conc_scheme_a}, we have created the indexing scheme using the chain order heuristic (a path decomposition), while in Table \ref{tables:CO_H3conc_scheme_b}, we use our chain decomposition algorithm.  Observe that even if the required time to compute the chain decomposition is higher, the total time is about the same, which shows the growing efficiency of the algorithms as the number of chains becomes smaller.



\begin{table}
     \centering
     \begin{subtable}[b]{0.9\textwidth}
         \centering
         \includegraphics[width=\textwidth]{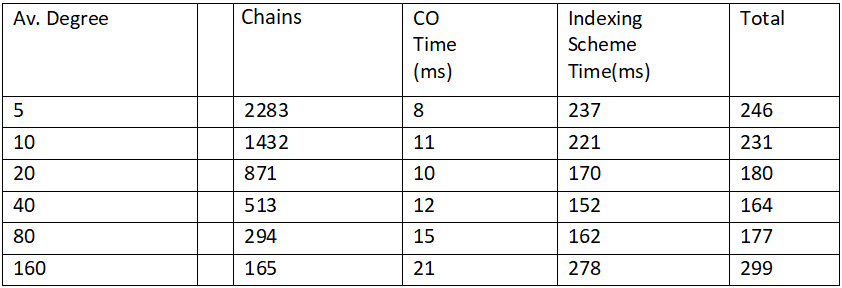}
         \caption{Metrics: Creating the indexing scheme in combination with the chain order heuristic.}
         \label{tables:CO_H3conc_scheme_a}
     \end{subtable}
     \hfill
     \begin{subtable}[b]{0.9\textwidth}
         \centering
         \includegraphics[width=\textwidth]{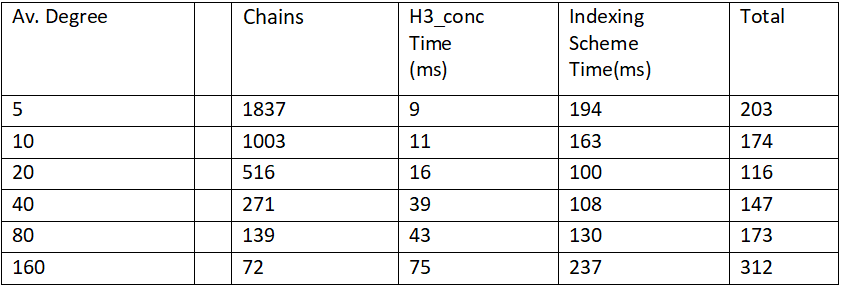}
         \caption{Metrics: Creating the indexing scheme in combination with algorithm \ref{alg:H3_conc} for chain decomposition.}
         \label{tables:CO_H3conc_scheme_b}
     \end{subtable}
        \caption{Run times of the indexing scheme using path and chain decomposition.}
        \label{tables:CO_H3conc_scheme}
\end{table}

Next, we performed experiments using the same graphs of 5000 and 10000 nodes as we described in Section \ref{ExperimentsChains} produced by the three different models of the Networkx~\cite{networkx} and the Path-Based model~\cite{lionakis2021constant}. We computed a chain decomposition using our  best performing approach, Algorithm \ref{alg:H3_conc}, H3\_conc, and created an indexing scheme using Algorithm \ref{alg:IndexingScheme}. For simplicity, we assume that the input graph has sorted adjacency lists, having ran the sorting of the adjacency lists, Algorithm \ref{alg:SortingAdjList}, as a preprocessing step. 
We report our experimental results in Tables \ref{table:IndexingSchemeResults5000} and \ref{table:IndexingSchemeResults10000} for graphs with 5000 nodes and graphs with 10000 nodes, respectively.  The meaning of the columns of the tables is as follows:
\begin{itemize}
  \item \textbf{Av. Degree}: The average degree of the graph.
  \item \textbf{Chains}: Number of chains computed by our heuristic (H3\_conc). 
  \item $|\textbf{E}_{\textbf{tr}}|$: Number of transitive edges.
  \item $|\textbf{E}_{\textbf{red}}|$: Number of non-transitive edges.
  \item $|\textbf{E}_{\textbf{tr}}|/|\textbf{E}|$: The percentage of transitive edges.
  \item \textbf{H3\_conc Time (ms)}: The running time for the chain decomposition step.
  \item \textbf{Indexing Scheme Time (ms)}: The running time for the indexing scheme creation step.
  \item \textbf{Total}: The total time(ms) needed to decompose the graph and create the indexing scheme. It is the sum of the two preceding cells.
  \item \textbf{TC}: The time needed to perform $n$ DFS (or BFS) starting from every vertex to mark the reachable vertices; the time complexity is $O(|V|*|E|)$.  The results are stored in a 2-dimensional adjacency matrix.
\end{itemize}

In theory, the phase of the indexing scheme creation needs $O(k_c*|E_{red}|+|E_{tr}|)$ steps. However, the numbers on the tables reveal some interesting findings:  (a) as the average degree increases and the graph becomes denser, the cardinality of $E_{red}$ remains almost stable; (b) the number of chains decrease; (c) we observe that the number of non-transitive edges, $E_{red}$, does not vary significantly as the average degree increases, i.e., the number of the transitive edges, $|E_{tr}|$, increases since $(E_{tr}=E-E_{red})$. Since the algorithm merely traces in linear time the transitive edges, the growth of $|E_{tr}|$ affects the run time only linearly. As a result, the run time of our technique does not increase as the the size of the (edges) input graph increases. To demonstrate it clearly, we show the curves of the running time in Figure \ref{fig:TimeErdosComp10000} for the graphs of 10000 nodes produced by the ER model. The flat (blue line) represents the run time of the indexing scheme, and the curve (red line) the run time of the DFS-based algorithm for computing the transitive closure (TC). Clearly, the time of the DFS-based algorithm  increases as the average degree increases, while the time of the indexing scheme is a straight line parallel to the $x$-axis. All models follow this pattern, see Tables \ref{table:IndexingSchemeResults5000} and \ref{table:IndexingSchemeResults10000}.
\begin{table}
\centering
\includegraphics[scale=0.8]{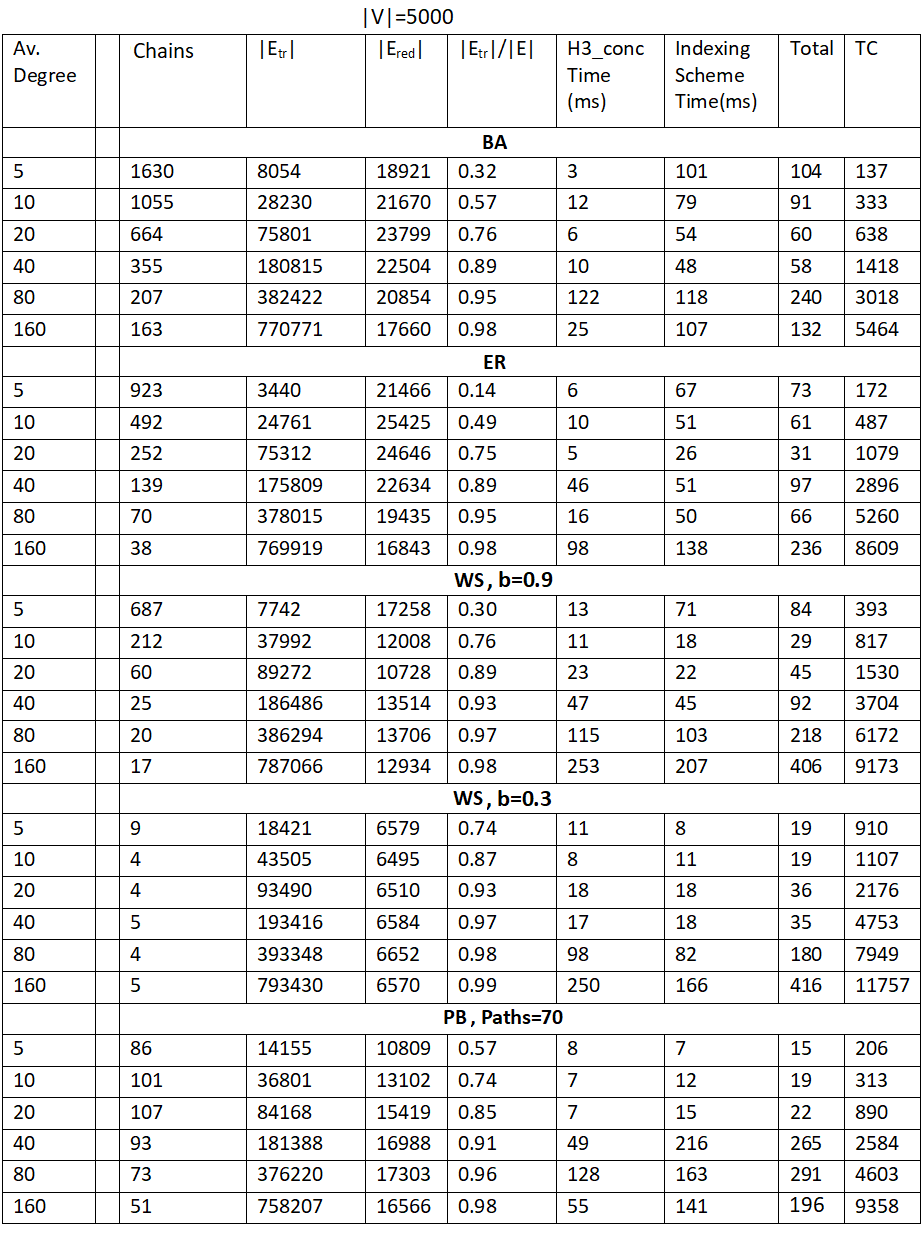}
\caption{Indexing scheme analysis on graphs of 5000 nodes.}
\label{table:IndexingSchemeResults5000}
\end{table}

\begin{table}
\centering
\includegraphics[scale=0.8]{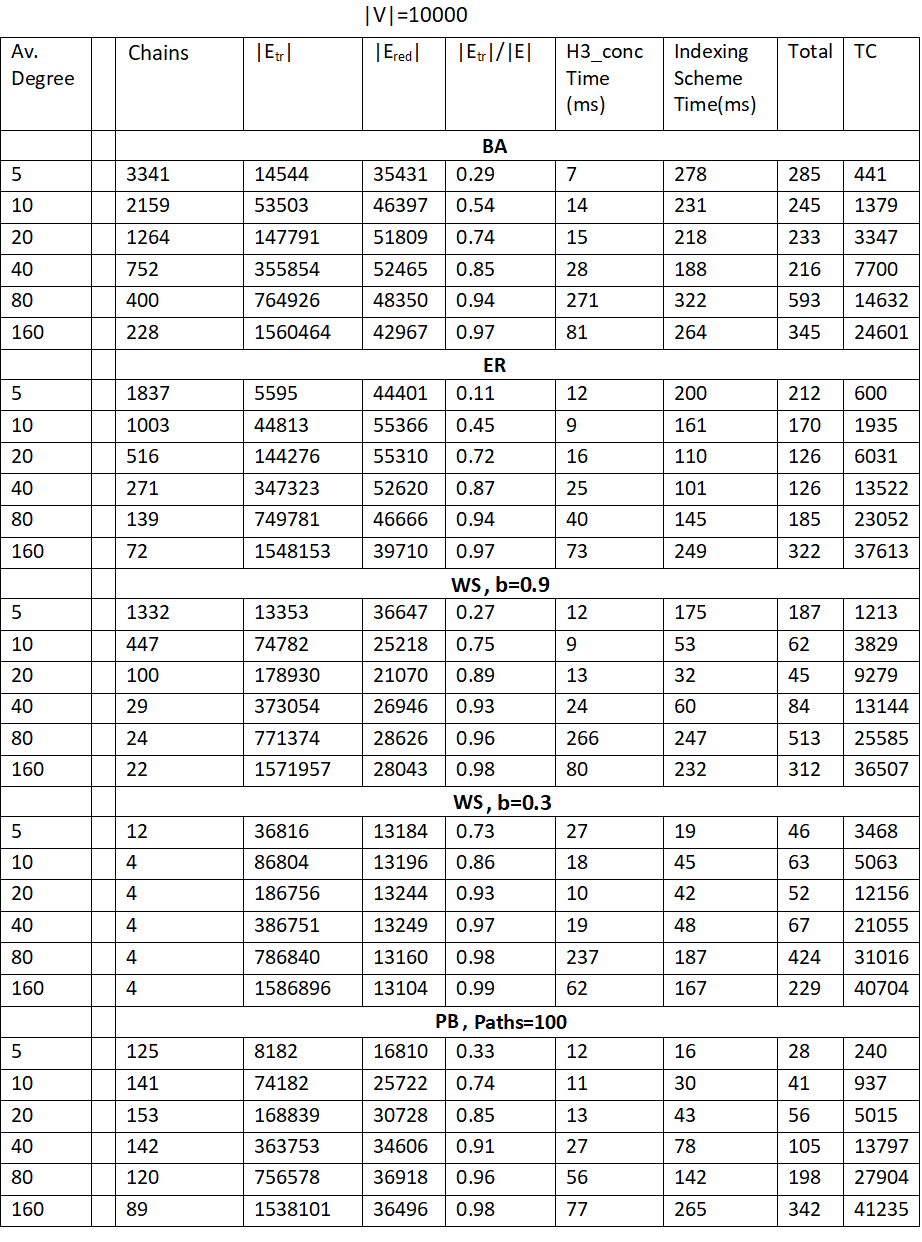}
\caption{Indexing scheme analysis on graphs of 10000 nodes.}
\label{table:IndexingSchemeResults10000}
\end{table}

We chose to decompose the graph into chains with our most efficient algorithm since it produces the fewest chains among all heuristics. Clearly, there is a trade-off to consider when building an indexing scheme. Assume that we have a path decomposition, and then we perform chain concatenation.  
If there is no concatenation between two paths, the concatenation algorithm will run in linear time.
On the other hand, if there are concatenations, for each of the paths, then the cost is $O(l)$ time, but the savings in the indexing scheme creation is $\Theta(|V|)$ in space requirements and $\Theta (|E_{red}|)$ in time, since every concatenation reduces the needed index size for every vertex by one. Hence, in comparison with doing merely a path decomposition, applying path concatenation, we create a more compact indexing scheme faster.
Another interesting and to some extent surprising observation that comes from the results of Tables \ref{table:IndexingSchemeResults5000} and \ref{table:IndexingSchemeResults10000} is that the transitive edges for almost all models of the graphs of 5000 and 10000 nodes with average degree above 20 is above 85\%, i.e.,  $|{E}_{{tr}}|/|{E}| \geq 85\%$, see the appropriate columns in both tables.  In some cases where the graphs are a bit denser the percentage grows above 95\%.  This has important implications in designing practical algorithms for faster transitive closure computation.

\begin{figure}
\centering
\includegraphics[scale=0.8]{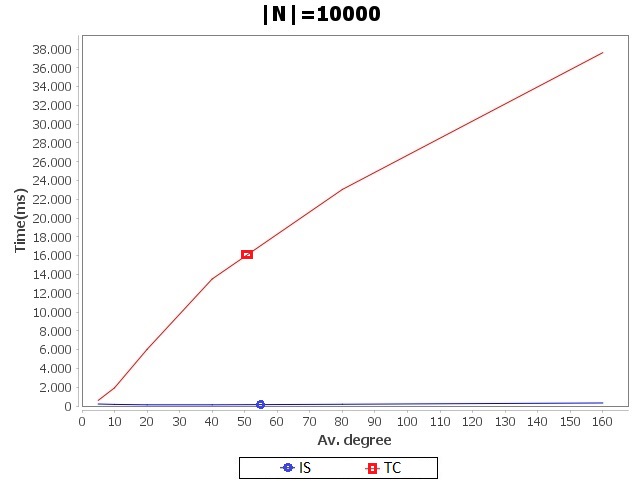}
\caption{Run time comparison between the Indexing Scheme (blue line) and TC (red line) for ER model on graphs of 10000 nodes, see Table \ref{table:IndexingSchemeResults10000}.}
\label{fig:TimeErdosComp10000}
\end{figure}

\section{Conclusions}
We presented heuristics that find a chain decomposition in almost linear time such that the number of chains is very close to the minimum and investigate how fast and practical chain decomposition algorithms can enhance transitive closure solutions. Our extensive experiments expose the practical behavior of (1) the width, (2) $E_{red}$, and (3) $E_{tr}$ as the density of graphs grows.  The also show the practical efficiency of our heuristics.
We show that the set $E_{red}$ is bounded  by $width*|V|$ and show how to find a substantially large subset of $E_{tr}$ in linear time given any path/chain decomposition. Finally, we build and evaluate an indexing scheme that allows us to answer reachability queries in constant time.  The time complexity to produce the scheme is $O(k_c*|E_{red}|)$, and its space complexity is $O(k_c*|V|)$.  
Our experimental work reveals the practical efficiency of this approach, especially for very large graphs with relatively high average degree.
Using recent results on parameterized linear time algorithms for computing the minimum number of chains, we showed that the transitive closure of a DAG can be computed in parameterized linear time.

The potential applications of these techniques in a dynamic setting where edges and nodes are inserted and deleted from a (very large) graph are significant in practice.  Although our techniques were not developed for the dynamic case, the picture that emerges is very interesting. According to our experimental results, see Tables \ref{table:IndexingSchemeResults5000} and \ref{table:IndexingSchemeResults10000}, the overwhelming majority of edges in a DAG are transitive.  The insertion or deletion of a transitive edge clearly requires a constant time update since it does not affect transitivity, and can be detected in constant time. On the other hand, the insertion or removal of a non-transitive edge may require a minor or major recomputation in order to reestablish a good chain decomposition.  However, even if the insertion/deletion of new nodes/edges causes  significant changes in the reachability index (transitive closure) one can simply recompute a chain decomposition in linear or almost linear time, and then recompute the reachability scheme in $O(k_c*|E_{red}|)$ time, and $O(k_c*|V|)$ space, which is still very efficient in practice. We plan to work on such dynamic indexes in the future.


\nocite{}
\printbibliography 
\end{document}